\documentclass[pra,twocolumn]{revtex4-2}
\usepackage{amsmath,amsfonts,amssymb,amsthm,bm}
\usepackage{bbm,nccmath}
\usepackage[framemethod=tikz]{mdframed} 
\usepackage{stmaryrd,mathrsfs,mathtools}
\usepackage{array} 
\usepackage{multirow}
\usepackage{wasysym}
\usepackage{enumitem}
\usepackage[most]{tcolorbox}
\usepackage{hyperref}
\usepackage{breakurl}
\usepackage{color}
\usepackage{tikz}
\usetikzlibrary{cd,calc}
\usepackage{titlesec}
\usepackage[leqno,fleqn,intlimits]{empheq}  
\usepackage{atbegshi}

\tikzset{
  symbol/.style={
    draw=none,
    every to/.append style={
      edge node={node [sloped, allow upside down, auto=false]{$#1$}}}
  }
}

\newcommand{\defeq}{ {\kern 0.2em}:{\kern -0.5em}={\kern 0.2 em} }  
\newcommand{\eqdef}{ {\kern 0.2em}={\kern -0.5em}:{\kern 0.2 em} }  

\newcommand{\wh}{\widehat}

\DeclareMathOperator{\Ran}{ran}  
\DeclareMathOperator{\Id}{\mathbbm{1}} 

\newcommand{\Cmplx}{{\mathbb C}}   

\DeclareMathOperator{\Span}{\mathrm{span}}

\newcommand{\inpr}{\ilinpr}
\newcommand{\ilinpr}[2]{\langle {#1} \,|\, {#2} \rangle}   
\DeclareMathOperator{\Tr}{\mathrm{Tr}}

\newcommand{\kk}{\mathsf{k}}
\newcommand{\hh}{\mathsf{h}}

\newcommand{\cH}{\mathcal{H}}
\newcommand{\cK}{\mathcal{K}}

\newcommand{\cV}{\mathcal{V}}
\newcommand{\cW}{\mathcal{W}}

\newcommand{\cM}{\mathcal{M}}

\newcommand{\gk}{\mathfrak{k}}
\newcommand{\gK}{\mathfrak{K}}
\newcommand{\gh}{\mathfrak{h}}
\newcommand{\gJ}{\mathfrak{J}}
\newcommand{\linfont}{\mathfrak}

\newcommand{\Blin}{\linfont{B}} 
 
\newcommand{\TC}{\Blin^1}
\newcommand{\HS}{\Blin^2}

\newcommand{\CP}{\ensuremath{\mathsf{CP}}}

\newcommand{\sB}{\mathscr{B}}

\newcommand{\sS}{\mathscr{S}}

\newcommand{\pam}{\overline}
 
\newcommand{\bpi}{{\bm{\pi}}}
\newcommand{\btheta}{{\bm{\theta}}}
\newcommand{\reduce}{\mathsf{Red}}

\theoremstyle{plain}
\newtheorem{thm}{Theorem}

\newtheorem{lem}{Lemma}
\newtheorem{prop}[thm]{Proposition}

\newtheorem*{lem*}{Lemma}
\newtheorem*{cor*}{Corollary}
\newtheorem*{thm*}{Theorem}
\newtheorem*{prop*}{Proposition}
\theoremstyle{definition}
\newtheorem*{definition*}{Definition}
\let\origmaketitle\maketitle
\def\maketitle{
  \begingroup
  \def\uppercasenonmath##1{} 
  \let\MakeUppercase\relax 
  \origmaketitle
  \endgroup
}
\begin{document}

\title{Existence of Kraus decomposition over countably-infinite-dimensional
  Hilbert spaces via a strongly-convergent direct process tomography algorithm}
  \author{Paul~E.~Lammert}
 \email{lammert@psu.edu}
\address{
  Department of Physics \\ 
Pennsylvania State University \\ University Park, PA 16802-6300}
\date{\today}
\begin{abstract}
An algorithm is presented for Kraus decomposition of a completely positive operator over separable (countably-infinite-dimensional) Hilbert spaces, together with an elementary proof that the generated sum convergences in strong-operator topology. This improves on the standard, nonconstructive, proof by fusing the abstract problem with practical process tomography. Kraus operators are generated one-by-one, each having one more guaranteed zero matrix entry than the previous one. In this way, the stream of outputs of the algorithm provides a coherent family of Kraus decompositions of restrictions of the target CP map to ever-larger subspaces.
\end{abstract}

\maketitle
\section{Introduction}

Existence of Kraus decompositions of $\CP$ maps
(also known as Kraus form\cite{Vacchini-Foundations}
or operator-sum representation\cite{Nielsen+Chuang,Lidar-20,Chruscinski-22})
is a fundamental result relevant to quantum information and open
quantum systems\cite{Nielsen+Chuang,
  Breuer+Petruccione,
  Benatti-09,
  Hayashi-17,
  Lidar-20,
  Chruscinski-22,
  Vacchini-Foundations}.
The statement is:
If $\Lambda$
is a completely positive operator in $\Blin(\TC(\cH),\TC(\cK))$,
with $\cH$ and $\cK$ separable Hilbert spaces,
then there are
sequences $A_1,\ldots$ of operators in $\Blin(\cH,\cK)$ such that
\begin{equation}\label{eq:Kraus}
\Lambda = \sum_{i=1}^\infty \btheta(A_i),  
\end{equation}
with
\begin{equation}
  \btheta(A) \defeq
  \;\rho\mapsto A\rho A^\dag,
\end{equation}
and the sum converges in strong-operator topology (SOT).
That is, for any $\rho$, $\sum_{i=1}^n \btheta(A_i)\rho$ converges
to $\Lambda\rho$ in trace-norm as $n\to\infty$.
Operators of form $\btheta(A)$ are the sole extreme, or pure, $\CP$ maps
in the sense that they are the ones which can be so decomposed only
into multiples of themselves. This explains the foundational interest.
Such decomposition is also of practical use as a convenient presentation
of noisy channels.
Process tomography\cite{Chuang+Nielsen-97,Jaeger-07}
extracts a suitable set of Kraus operators $A_i$
from measurements of transition probabilities.

Proof of existence of Kraus decompositions
in finite dimension is well-known\cite{Nielsen+Chuang,Breuer+Petruccione,Hayashi-17}.
However, it does not suggest the more general case.
Indeed, if spectral decomposition of hermitian operators is any guide,
one might well expect that integrals, rather than sums, are generally
required.
Of course, Kraus gave a proof\cite{Kraus-71,Kraus-83}
of the decomposition (\ref{eq:Kraus}) for separable Hilbert spaces.
Yet it is unsatisfying in at least two respects:
it relies on sophisticated results from representations of
$C^*$-algebras, and it is nonconstructive, giving not
the slightest hint how one might go about actually
finding such a decomposition or an approximation thereof.

The approach taken here does not have those shortcomings.
We present an algorithm to actually construct a Kraus decomposition,
accompanied by an elementary convergence analysis.
In the terminology of Ref.~\cite{Mohseni+Rezakhani+Lidar-08}, this is
\textit{direct} process tomography.
The only inputs are the object $\CP$ map $\Lambda$ in the form of a
capacity to obtain individual ground matrix elements,
and orthonormal bases of $\cH$ and $\cK$.
Here is a brief description; full details are contained in
Section \ref{sec:algorithm}.
Starting with $\Lambda_0$, we generate sequences
$A_1,A_2,\ldots$ of Kraus operators and $\Lambda_1,\Lambda_2,\ldots$
of remainder $\CP$ maps in such a way that 
$\Lambda_n = \btheta(A_{n+1}) + \Lambda_{n+1}$.
In the $n$-th step, $\Lambda_n$ is ``reduced'' with respect to
a new pair $(\kk_j,\hh_i)$ from the given orthonormal bases.
This means that there is a set of
guaranteed zeros in the four-dimensional array of basic matrix
elements of $\Lambda_n$ which grows with $n$ according to the
reduced pairs. (Section \ref{sec:matrix} clarifies this.)
That the procedure achieves what is claimed, namely Eq. (\ref{eq:Kraus}),
requires proving that those $A_n$ matrices actually represent
bounded operators, and that $\Lambda_n\xrightarrow{\text{SOT}} 0$.
These proofs are carried out in Section \ref{sec:proofs}.
That the $A_n$ matrices have more guaranteed zero entries
the larger $n$ is (Sec. \ref{sec:matrix}),
and that the remainders $\Lambda_n$ are $\CP$,
play important r\^{o}les there.
These facts are also of great practical interest. They imply that
intermediate results $\Lambda = \sum_{i=1}^N  \btheta(A_i) + \Lambda_N$
carry all information about the restriction (better: \textit{projection})
of $\Lambda$ to appropriate subspaces in the finite sum, with a
uniformly bounded $\CP$ remainder.
In fact, 
$\|A\|^2, \|\Lambda_n\|_{1,1} \le \|\Lambda\|_{1,1}$,
where
\begin{equation}\nonumber
\|\Gamma\|_{1,1}
\defeq \sup \{ |\Gamma\rho\|_1 \,|\, \|\rho\|_1=1 \}
\end{equation}
denotes the operator norm on \mbox{$\Blin(\TC(\cH),\TC(\cK))$}, and
\mbox{$\|\cdot\|_1$}, trace norm.

\section{Reduction algorithm}\label{sec:algorithm}

This section explains the algorithm. Proof of Prop.~\ref{prop:reduction step}
and SOT convergence of the algorithm are deferred to
Section \ref{sec:proofs}.
Some tools and notation are explained in Section \ref{sec:tools}.
The reduction step and its properties are described in
Section \ref{sec:step}, and extended into a complete algorithm in
Section \ref{sec:iteration}.
Section \ref{sec:matrix} describes what the procedure
looks like from a matrix perspective.

\subsection{Preliminaries}\label{sec:tools}

\subsubsection{Ground matrix elements}

First, we establish some notation\cite{Lammert-26}.
We deal with two Hilbert spaces $\cH$ and $\cK$ and associated
(super)operators.
denote vectors in $\cH$ and $\cK$.
This makes clear which Hilbert space vectors belong to and the font
distinguishes from other possible uses of these letters.
$\pam{\hh}$ is the dual vector canonically and conjugate-linearly related to $h$:
$\pam{\hh}(\hh') = \inpr{\hh}{\hh'}$, but we prefer the inner product
notation. $\bpi(\hh) = \hh\pam{\hh}$ is then $\|\hh\|^2$ times the orthogonal
projector onto the span of $\hh$.
We mostly deal with a bounded linear operator $A\in\Blin(\cH,\cK)$ in
the form of \textit{matrix elements} $\inpr{\kk}{A\hh}$. This \textit{does not}
imply that $\kk$ or $\hh$ belong to any particular basis
(there isn't really a matrix).
Over finite-dimensional spaces, $\Blin(\cH,\cK)$ becomes a Hilbert space
$\HS(\cH,\cK)$ under the Hilbert-Schmidt inner product
$\inpr{A}{B} = \Tr A^\dag B$. This cannot be done in infinite dimensions,
but even in that case it is perfectly safe to pair bounded operators and
trace-class operators in this way, so we will use that notation, being
careful to observe the restriction.

Density operators over $\cH$ belongs to the subset
$\TC(\cH)^+$ of positive (indicated by the `$+$' superscript)
trace-class operators on $\cH$.
The convenient way to write an eigenfunction expansion of a density
operator in our notation is $\textstyle{\sum} \bpi(\hh_i)$.

Next we come to $\Blin(\TC(\cH),\TC(\cK))$,
the bounded operators from trace-class to trace-class,
of which the $\CP$ maps comprise a subclass $\CP(\cH,\cK)$.
$\Lambda$ in $\Blin(\TC(\cH),\TC(\cK))$ is fully described by its
\textit{ground matrix elements} $\inpr{\kk'\pam{\kk}}{\Lambda\cdot \hh'\pam{\hh}}$ bearing in mind what we said above about the inner product notation.
(The dot there is just a separator to aid in parsing.
The only possible parse of $\Lambda \hh'\pam{\hh}$ is
$\Lambda(\hh'\pam{\hh})$ since $(\Lambda\hh')\pam{\hh}$ makes no sense.)

\subsubsection{Kernel relations}

We now introduce a device, \textit{kernel relations} 
to keep track of which (ground) matrix elements are zero.
For an operator $A\in\Blin(\cH,\cK)$ define 
\begin{equation}
\gK_A \defeq \{(\kk,\hh)\,|\, \inpr{\kk}{A\hh} = 0\}.
\end{equation}
It is easy to see that $\{\kk\,|\, (\kk,\hh)\in\gK_A\}$ is a closed subspace of
$\cK$ and $\{\hh\,|\, (\kk,\hh)\in\gK_A\}$ of $\cH$.
The set of subsets of $\cK\times\cH$ with this property is denoted $\sS$.
The kernel relation for a superoperator $\Lambda\in\Blin(\TC(\cH),\TC(\cK))$
is defined similarly: with $\bpi(\hh) \defeq \hh\pam{\hh}$,
\begin{equation}\label{eq:K of Lambda}
  \gK_\Lambda \defeq \{(\kk,\hh)\,|\, \inpr{\bpi(\kk)}{\Lambda\cdot\bpi(\hh)} = 0\}.
\end{equation}
To understand the significance of this, note that
\mbox{$\Lambda\in\CP(\cH,\cK)$} satisfies the
quasi-Cauchy-Schwarz inequality
\begin{equation}\label{eq:quasi-CS}
|\inpr{\kk'\pam{\kk}}{\Lambda\cdot \hh'\pam{\hh}}|^2  
\le
\inpr{\bpi(\kk')}{\Lambda\cdot \bpi(\hh')}
\inpr{\bpi(\kk)}{\Lambda\cdot \bpi(\hh)}.
\end{equation}
This is proven in Section \ref{sec:quasi-CS}.
The quasi-CS inequality together with polarization imply that 
$\gK_\Lambda$ belongs to $\sS$ if $\Lambda$ is $\CP$.
Not coincidentally, $\gK_{\btheta(A)} = \gK_A$.
For $U,V \subseteq \cK\times\cH$,
$U \vee_S V$ denotes the smallest relation in $\sS$ which contains
$U\cup V$, namely,
\mbox{$\bigcap \{W\in S\,|\, U\cup V\subseteq W\}$}.

\subsection{Basic reduction step}\label{sec:step}

\begin{prop}\label{prop:reduction step}
  Given $\Lambda\in\CP(\cH,\cK)$, $h\in\cH$ and \mbox{$k\in\cK$}
  such that $(\kk,\hh)\not\in \gK_\Lambda$,
there exists \mbox{$0 \neq A\in\Blin(\cH,\cK)$},
unique up to a phase factor, such that
$\Lambda' \defeq \Lambda - \btheta(A)$
satisfies $(\kk,\hh) \in \gK_{\Lambda'}$.
Namely,
  \begin{equation}\label{eq:A}
  \inpr{\kk'}{A\hh'} = \frac{\inpr{\kk'\pam{\kk}}{\Lambda\cdot \hh'\pam{\hh}}}
                        {\sqrt{\inpr{\bpi(\kk)}{\Lambda\cdot \bpi(\hh)}}}.
  \end{equation}
In addition,
\begin{subequations}\label{eq:new kernels}
  \begin{align}
 \gK_{A} & \supseteq \gK_\Lambda \label{eq:A kern} \\
\Lambda' &\in \CP(\cH,\cK) \label{eq:Lambda' is CP} \\
\gK_{\Lambda'} & \supseteq \gK_\Lambda \vee_{S} (\kk,\hh) \label{eq:Lambda' kern}
  \end{align}
\end{subequations}

\end{prop}

\subsection{Iteration}\label{sec:iteration}

Now, let
\begin{itemize}[label=$\circ$]
\item  
  $\sB_\cH = \hh_1,\hh_2,\ldots$ be an ONB for $\cH$
\item $\sB_\cK = \kk_1,\kk_2,\ldots$ be an ONB for $\cK$
\item $n\mapsto(\gk(n),\gh(n))$ be an enumeration of $\sB_\cK\times\sB_\cH$.
\end{itemize}
The idea is to perform a reduction step with each $(\gk(n),\gh(n))$ in turn.
Define
\begin{equation}
  \reduce(\kk,\hh,\Lambda) \defeq
  \begin{cases}
0, &     (\kk,\hh)\in\gK_\Lambda \\
\text{$A$ in (\ref{eq:A})}, &  \text{otherwise}.
  \end{cases}
\end{equation}
Take $\Lambda_0 \defeq \Lambda$ and for $m \ge 1$,
\begin{subequations}
  \begin{align}
  A_m & \defeq \reduce(\gk(m),\gh(m),\Lambda_{m-1}) \\
  \Lambda_m & \defeq \Lambda_{m-1} - \btheta(A_m)
  \end{align}
\end{subequations}
Thus,
\begin{equation}\label{eq:reduction n again}
\Lambda_0 = \sum_{i=1}^{n} \btheta(A_i) + \Lambda_n.
\end{equation}
Below, we prove that
\mbox{$\sum_{s=1}^n \btheta(A_s) \xrightarrow{\mathrm{SOT}}\Lambda$}
as $n\to\infty$;
equivalently, that $\Lambda_n\xrightarrow{\mathrm{SOT}} 0$.
First, we briefly consider the finite-dimensional approximations this
algorithm provides.

\subsection{The matrix view}
\label{sec:matrix}

Properties of the (putative) Kraus operators $A_n$ are easily discussed,
and visualized, in terms of matrices. One way to deal with
the more complicated nature of the $\Lambda_n$ is exemplified by
(\ref{eq:K of Lambda}). A description as matrices of matrices is also
useful: think of $\inpr{\kk_i\pam{\kk_j}}{\Lambda\cdot\hh_m\pam{\hh_n}}$
as the $(i,m)$ entry of the $(j,n)$ entry.
With this setup, we have a simple and unified description of the way
the matrices develop guaranteed zeros. 
The $(\gk(m),\gh(m))$ entries of $A_n$ and $\Lambda_n$ are
guaranteed zero as soon as $m\le n$. For $\Lambda_n$, this means
at both levels, i.e., the ``supermatrix'' has zero matrices,
and \textit{all} of its matrices have ordinary zeros, at those positions.

To make it even simpler, and more practical, consider a sensible
reduction schedule such as in the
following diagram: the numbers indicate the order in which
to take pairs $(\gk(m),\gh(m))$ in the reduction algorithm.
For instance, here, $\gh(8) = \hh_3$, $\gk(8) = \kk_2$.
\begin{center}
\begin{tikzpicture}[scale=0.75]
  \coordinate (shift) at (0,0); 
  \draw (0,5)+(shift) node{1};
  \draw (1,5)+(shift) node{2};
  \draw (0,4)+(shift) node{3};
  \draw (2,5)+(shift) node{4};
  \draw (1,4)+(shift) node{5};
  \draw (0,3)+(shift) node{6};
  \draw (3,5)+(shift) node{7};
  \draw (2,4)+(shift) node{8};
  \draw (1,3)+(shift) node{9};
  \draw (0,2)+(shift) node{10};
  \draw (4,5)+(shift) node{11};
  \draw (3,4)+(shift) node{12};
  \draw (2,3)+(shift) node{13};
  \draw (1,2)+(shift) node{14};
  \draw (0,1)+(shift) node{15};
  \draw (5,5)+(shift) node{16};
  \draw (4,4)+(shift) node{17};
  \draw (3,3)+(shift) node{18};
  \draw (2,2)+(shift) node{19};
  \draw (1,1)+(shift) node{20};
  \draw (0,0)+(shift) node{21};
  \coordinate (dotstart) at (3.2,1.8);
  \draw[fill] (0,0) + (dotstart) circle(1.5pt);
  \draw[fill] (0.3,-0.3) + (dotstart) circle(1.5pt);
  \draw[fill] (0.6,-0.6) + (dotstart) circle(1.5pt);
  \coordinate (l) at (-0.85,0);
  \draw (0,0) + (l) node{$\kk_6$};
  \draw (0,1) + (l) node{$\kk_5$};
  \draw (0,2) + (l) node{$\kk_4$};
  \draw (0,3) + (l) node{$\kk_3$};
  \draw (0,4) + (l) node{$\kk_2$};
  \draw (0,5) + (l) node{$\kk_1$};
  \coordinate (b) at (0,5.8);
  \draw (0,0) + (b) node{$\hh_1$};
  \draw (1,0) + (b) node{$\hh_2$};
  \draw (2,0) + (b) node{$\hh_3$};
  \draw (3,0) + (b) node{$\hh_4$};
  \draw (4,0) + (b) node{$\hh_5$};
  \draw (5,0) + (b) node{$\hh_6$};
  \draw[dashed] (-0.4,-0.5) -- ++(0,5.9) -- ++(6,0);
\end{tikzpicture}
\end{center}
As $n$ increases, there is a larger and larger
upper-left block of guaranteed zeros in the matrix of $A_n$,
the matrix of $\Lambda_n$, and the matrices comprising the entries of that.

\section{Deferred proofs}\label{sec:proofs}

Section \ref{sec:proof of reduction} gives the proof of
Prop.~\ref{prop:reduction step}. Using some preliminary ideas
from Section \ref{sec:proof prelims}, the proposition is reduced
to an elementary result about positive operators on Hilbert spaces
which is given in Section \ref{sec:Pos red}.
Section \ref{sec:proof Lambda_n tends to zero} fills in the last
missing piece of the full algorithm, namely 
$\Lambda_n\xrightarrow{\mathrm{SOT}} 0$.

\subsection{Preliminaries}\label{sec:proof prelims}
%
\subsubsection{$\gJ$-transform}

Assume that $\dim\cH, \dim\cK < \infty$.
Then, the \mbox{$\gJ$-transform}
of $\Lambda\in\Blin(\Blin(\cH),\Blin(\cK))$
can be defined via ground matrix elements as
\begin{equation}\label{eq:Jam}
\inpr{\kk'\pam{\hh'}}{\gJ\Lambda\cdot \kk\pam{\hh}}
= \inpr{\kk'\pam{\kk}}{\Lambda\cdot \hh'\pam{\hh}}.
\end{equation}
The important point is
\begin{equation}\label{eq:J correspondence}
\gJ\Lambda\in\Blin(\HS(\cH,\cK))^+ \;\Leftrightarrow\;
\Lambda\in\CP(\cH,\cK),
\end{equation}
i.e., $\gJ\Lambda$ is a positive operator on the Hilbert space
$\HS(\cH,\cK)$ (Hilbert-Schmidt inner product) exactly when $\Lambda$
is $\CP$.
$\gJ$ is a basis-free variant of the Choi-Jamio{\l}kowski isomorphism
and (\ref{eq:J correspondence}) is the corresponding version of
channel-state duality\cite{%
Choi-75,
Nielsen+Chuang,
Benatti-09,
Jiang+Luo+Fu-13,
Hayashi-17,
Vacchini-Foundations}.
For proof and thorough discussion,
see Refs.~\cite{Lammert-26,Grabowski+Kus+Marmo-07}.
A useful special case of (\ref{eq:Jam}) is
\begin{equation}\label{eq:Jam-pi}
\inpr{\bpi(\kk)}{\Lambda\cdot \bpi(\hh)}
= \inpr{\kk\pam{\hh}}{\gJ\Lambda\cdot \kk\pam{\hh}}.
\end{equation}
%
\subsubsection{Lifting projections}

For discussing restrictions to finite-dimensional spaces in a compact way,
we introduce the notion of a \textit{lifted polyprojection}.
By $P$ we denote a generic \textit{polyprojection};
it combines orthoprojections on $\cH$ and $\cK$.
In the expression $P\hh$ or $P\kk$ the appropriate projection is applied.
The lifting of $P$, denoted by a hat, 
operates on (super)operators as follows:
\begin{equation}
  \wh{P}A = P \circ A \circ P, \quad
  \wh{P}\Lambda = \wh{P} \circ \Lambda \circ \wh{P}. 
\end{equation}
Again, the appropriately typed instantiation
is picked out for each occurence of $P$ or $\wh{P}$.

\subsubsection{Quasi-Cauchy-Schwarz inequality}\label{sec:quasi-CS}

Recall, the inequality (\ref{eq:quasi-CS}) in question is ($\Lambda$ is $\CP$)
\begin{equation}\nonumber
|\inpr{\kk'\pam{\kk}}{\Lambda\cdot \hh'\pam{\hh}}|^2  
\le
\inpr{\bpi(\kk')}{\Lambda\cdot \bpi(\hh')}
\inpr{\bpi(\kk)}{\Lambda\cdot \bpi(\hh)}.
\end{equation}
  Assume first that $\dim\cH,\dim\cK<\infty$.
  
  The Cauchy-Schwarz inequality applies to any positive hermitian form, e.g.,
  for a positive operator $T$, $|\inpr{x}{Ty}|^2 \le |\inpr{x}{Tx}\inpr{y}{Ty}|$,
  and $T$ here could be $\gJ\Lambda$.
  Thus,
\begin{equation}\nonumber
  \begin{split}
|\inpr{\kk'\pam{\kk}}{\Lambda\cdot \hh'\pam{\hh}}|^2  
& =
| \inpr{\kk'\pam{\hh'}}{\gJ\Lambda\cdot \kk\pam{\hh}}|^2 \\
& \le
|\inpr{\kk\pam{\hh}}{\gJ\Lambda\cdot \kk\pam{\hh}}
\inpr{\kk'\pam{\hh'}}{\gJ\Lambda\cdot \kk'\pam{\hh'}}| \\
& =
|\inpr{\kk'\pam{\kk'}}{\Lambda\cdot \hh'\pam{\hh'}}
\inpr{\kk\pam{\kk}}{\Lambda\cdot \hh\pam{\hh}}|.
\end{split}
\end{equation}
To lift the dimensionality restriction, we need merely note that
(i) the inequality is true for $\Lambda$ as soon as it is true for
$\wh{P}\Lambda$, where $P$ is a polyprojection
with $\kk,\kk',\hh,\hh'$ in its
range, and (ii) $\wh{P}\Lambda$ is $\CP$.

\subsection{Reducing a positive operator}\label{sec:Pos red}

As is explained in detail below,
the $\gJ$-transform allows Prop.~\ref{prop:reduction step} to be deduced
from the following simple lemma.
\begin{lem}\label{lem:Pos red}
  Given $T\in\Blin(\cM)^+$, $\cM$ a Hilbert space $\cM$, and
  $\varphi\in\cM$ such that $T\varphi\neq 0$,
there exists unique $\eta$ (modulo a phase factor) such that
\begin{equation}\label{eq:T' def}
T' \defeq T - \bpi(\eta)
\end{equation}
  satisfies
\begin{equation}
T'\varphi = 0 \label{eq:T'(varphi)=0}.
\end{equation}
Namely,
\begin{equation}\label{eq:eta}
    \eta =
\frac{T{\varphi}}{\sqrt{\inpr{\varphi}{T\varphi}}}.
  \end{equation}
Furthermore, 
  \begin{equation}
  0 \le  T' \le T \label{eq:T' pos}
\end{equation}
  and
  \begin{equation}
  \ker T'  = \ker T + \Cmplx \varphi \label{eq:ker T'}
  \end{equation}
\end{lem}
\begin{proof}
(\ref{eq:T' def},\ref{eq:T'(varphi)=0}) imply 
$\eta \inpr{\eta}{\varphi} = T\varphi$. With appropriate
choice of phase, (\ref{eq:eta}) follows after minor rearrangement.
For (\ref{eq:T' pos}), clearly $T'\le T$. To see $0\le T'$,
consider
\begin{equation}\nonumber 
  \begin{split}
 \inpr{\varphi}{T\varphi} \inpr{\psi}{\bpi(\eta)\psi}
& \stackrel{(\ref{eq:eta})}{=} |\ilinpr{\psi}{T \varphi}|^2 
\\ & \le \inpr{\varphi}{T \varphi} \inpr{\psi}{T \psi},
  \end{split}
\end{equation}
for any $\psi$, where the inequality is Cauchy-Schwarz.
Therefore, \mbox{$\bpi(\eta) \le T$}, and $0\le T'$ follows.

By positivity, $T \psi = 0 \Leftrightarrow \inpr{\psi}{T \psi} = 0$ and
similarly for $T'$. Also,
\begin{equation}\nonumber
  0 \le \inpr{\psi}{T' \psi}
  =
 \inpr{\psi}{T \psi} - |\inpr{\psi}{\eta}|^2.
\end{equation}
So, $\ker T' \supseteq \ker T$, but since $T'\varphi = 0$ by
construction, (\ref{eq:ker T'}) follows with ``$=$'' replace by ``$\supseteq$''.
The reverse inclusion is established by a form of dimension counting.
We have $\ker T' = \ker T \oplus \cV$ (\textit{orthogonal} direct sum),
where $\dim \cV \ge 1$, since $\varphi\not\in\ker T$.
On the other hand,
$\Ran T = \Ran(T'+\bpi(\eta)) = \Ran T' + \Cmplx \eta
= \Ran T' \oplus \cV'$, where $\dim \cV'\le 1$.
Therefore,
\begin{equation}
\begin{aligned}
  \cM = \Ran T & \oplus \ker T  &= \Ran T' \oplus \ker T \oplus \cV'
        \nonumber
\\ = \Ran T' &\oplus \ker T'  &= \Ran T' \oplus \ker T \oplus \cV.    
        \nonumber
\end{aligned}
\end{equation}
So, $\cV=\cV'$ is one-dimensional.
\end{proof}
\subsection{Proof of reduction step (Prop.~\ref{prop:reduction step})}
\label{sec:proof of reduction}

\subsubsection{Finite dimension}

If $\cH$ and $\cK$ are finite-dimensional, then
$\gJ\Lambda$ is well-defined as a positive operator on
the Hilbert space $\HS(\cH,\cK)$.
Using $\gJ$, it is straightforward to translate the
particular instantiation of Lemma~\ref{lem:Pos red}
given by
\begin{equation}\nonumber
  \begin{split}
    \cM & = \HS(\cH,\cK),\\
    T & = \gJ\Lambda,\\
    \varphi & = \kk\pam{\hh}, \\
    \eta & = A, \\
    T' &= \gJ\Lambda'
  \end{split}
\end{equation}
into Prop.~\ref{prop:reduction step}.
To translate (\ref{eq:ker T'}) into the kernel
relation (\ref{eq:Lambda' kern}), use (\ref{eq:Jam-pi}).
The point here is that
\begin{equation}\nonumber
\gK_\Lambda = \left\{ (\kk,\hh)\,:\, \kk\pam{\hh}\in\ker \gJ\Lambda\right\},
\end{equation}
the rank-1 part of $\ker\gJ\Lambda$,
while the application of Lemma \ref{lem:Pos red} refers to the full kernel.
For example, we may have $\kk_1\pam{\hh}_1 + \kk_2\pam{\hh}_2$
in the kernel of $\gJ \Lambda$, but neither
$\kk_1\pam{\hh}_1$ nor $\kk_2\pam{\hh}_2$ in $\gk_\Lambda$. 
If, through one step of the reduction algorithm, we then construct
$\gJ\Lambda'$ to have $\kk_1\pam{\hh}_1\in\ker\gJ\Lambda'$, this will
result in both $(\kk_1,{\hh}_1)$ and $(\kk_2,{\hh}_2)$ being in
$\gk_{\Lambda'}$.

\subsubsection{Countably infinite dimension}

First, we establish that $\reduce(\kk,\hh,\Lambda)$ is actually
a bounded operator. From the formula (\ref{eq:A}) for $A$'s matrix
and the quasi-CS inequality (\ref{eq:quasi-CS}), an explicit bound
\begin{equation}\nonumber
  \|A\|
  = \sup\{|\inpr{\kk'}{A\hh'}| \,:\, \|\kk'\|,\|\hh'\|\le 1\}
  \le \sqrt{ \|\Lambda_n\|_{1,1} }
\end{equation}
is immediately derived.

Now, because we work with matrix elements in Prop.~\ref{prop:reduction step}
and the kernel relations (\ref{eq:new kernels}) implicate only
restrictions to finite-dimensional subspaces, their validity follows
from the case just analyzed.
Specifically, if $P$ is a finite-dimensional polyprojection with
both $\hh$ and $\kk$ in its range,
$\wh{P}A = \reduce(\kk,\hh,\wh{P}\Lambda)$ and
$\wh{P}\Lambda' = \wh{P}\Lambda-\btheta(\wh{P}A)$.
By the just-proven finite-dimensional case, 
$\wh{P}\Lambda'$ is $\CP$. 
It remains to show that $\Lambda'$ itself is $\CP$.

Argue via a more general situation.
Let $\cV$ be a separable Hilbert space, $\rho\in\TC(\cV)^+$, and
$P_1,\ldots$ a sequence of projections onto nested finite-dimensional
subspaces $\cV_\sigma$ such that $\bigcup\cV_\sigma$ is dense in $\cV$;
The $\cH_s$ considered earlier are a special case.
Then, $\wh{P_\sigma}\rho\xrightarrow{\|\cdot\|_1}\rho$ as $\sigma\to\infty$.
[Proof:
For the special case $\rho=\bpi(v)$, we have
$\|\bpi(v)-\bpi(P_\sigma v)\|_1 \le 3 \|v\|\|P_\sigma^\perp v\|$,
but $\|P_\sigma^\perp v\|\to 0$.
For the general case, $\rho = \sum\bpi(v_i)$,
large terms of $P_\sigma \rho P_\sigma$
can be controlled by the above bound and small terms by
$\|\bpi(v)-\bpi(P_\sigma v)\|_1 \le 2 \|\bpi(v)\|_1$.]
Consequently,
\begin{equation}\label{eq:0 le rho}
  (\forall\sigma,\, 0 \le \wh{P_\sigma}\rho)
  \;\Rightarrow\;
  0\le\rho.
\end{equation}

Now, let $\cW$ be a second
separable Hilbert space with a similar sequence $(P_\alpha)$ of projections.
The following Lemma finishes the proof of Prop.~\ref{prop:reduction step}.
\begin{definition*}
For $\Upsilon\in\Blin(\TC(\cH),\TC(\cK))$, $0\le \Upsilon$ means that
  $\Upsilon$ maps $\TC(\cH)^+$ into $\TC(\cK)^+$,
  $\Upsilon \le \Psi$ that \mbox{$0 \le \Psi-\Upsilon$}.
\end{definition*}
\begin{lem*}
\begin{subequations}
  \begin{align}
  \label{eq:pos if projs are}
(\forall\alpha,\sigma,\, 0 \le \wh{P_{\alpha,\sigma}}\Gamma)
    &\;\Rightarrow\; 
         0 \le \Gamma. \\
\label{eq:CP if projs are}
(\forall\alpha,\sigma,\, \wh{P_{\alpha,\sigma}}\Gamma \text{ is } \CP)
& \;\Rightarrow\; \Gamma \text{ is }\CP.
  \end{align}
\end{subequations}
\end{lem*}
\begin{proof}
  
Given $0\le\rho$,
$(\Gamma\circ \wh{P_\sigma})\rho = \Gamma(\wh{P_\sigma}\rho)
\xrightarrow{\|\cdot\|_1} \Gamma\rho$, by (\ref{eq:0 le rho})
and boundedness of $\Gamma$. Hence,
\begin{equation}\nonumber
  (\forall\sigma,\, 0 \le \Gamma\circ\wh{P_\sigma})
  \;\Rightarrow\;
0 \le \Gamma.
\end{equation}
A very similar argument gives
\begin{equation}\nonumber
  (\forall\alpha,\, 0 \le \wh{P_\alpha}\circ\Gamma)
  \;\Rightarrow\;
0 \le \Gamma.
\end{equation}
Since
$\wh{P_{\alpha,\sigma}}\Gamma = \wh{P_\alpha}\circ\Gamma\circ\wh{P_\sigma}$,
the immediately preceding combine to yield (\ref{eq:pos if projs are}).

Now, by definition, $\Upsilon\in\CP(\cV,\cW)$ means that
\begin{equation}\label{eq:CP def}
\forall N \ge 1, \; 0 \le \Upsilon\otimes\Id_{\Blin(\Cmplx^N)}.
\end{equation}
Suppose $\Gamma\in\Blin(\Blin(\cH),\Blin(\cK))$ is such that
for every $\alpha,\sigma$,
$\wh{P_{\alpha,\sigma}}\Gamma$ is $\CP$.
So, (\ref{eq:CP def}) holds for 
$\Upsilon \defeq \wh{P_{\alpha,\sigma}}\Gamma$.
By (\ref{eq:pos if projs are}) with 
$\cV$ replaced by $\cV\otimes\Cmplx^N$,
$P_\sigma$ by $P_\sigma\otimes \Id_{\Cmplx^N}$, and similarly
with $\cW$ and $P_\alpha$, it follows that (\ref{eq:CP def})
holds with $\Upsilon\defeq \Gamma$.
That is, (\ref{eq:CP if projs are}) holds.
\end{proof}

\subsection{Proof that $\Lambda_n\xrightarrow{\mathrm{SOT}} 0$}
\label{sec:proof Lambda_n tends to zero}

Because $0\le \btheta(A_i)$, $0 \le \Lambda_n \le \Lambda$, 
Lemma \ref{lem:Lambda'<Lambda} below implies that
\mbox{$\| \Lambda_n \|_{1,1}\le \|\Lambda\|_{1,1}$} for every $n$.
This uniform bound implies that
$\|\Lambda_n\rho\|\to 0$ holds for every $\rho$
as soon as it holds for every $\rho$ in some set dense in $\TC(\cH)^+$.
$D=\cup_{s=1}^\infty \{\rho\,|\, 0\le \rho = \wh{P_s}\rho\}$
is such a set
because $\rho=\sum_{\alpha=1}^\infty\bpi(\hh_\alpha)$ can be arbitrarily well
trace-norm-approximated by $\sum_{\alpha=1}^n \bpi(\hh_\alpha)$, 
and (for $\alpha=1,\ldots,n$) each $\bpi(\hh_\alpha)$
by some $\wh{P_s}\bpi(\hh_\alpha) = \bpi(P_s\hh_\alpha)$.
Lemma \ref{lem:Lambda rho -> 0} thus completes the proof.

\begin{lem}\label{lem:Lambda'<Lambda}
$  0 \le \Gamma' \le \Gamma
  \;\Rightarrow\;
  \|\Gamma'\|_{1,1} \le \|\Gamma\|_{1,1}$.
\end{lem}
\begin{proof}
The hypothesis means that for every $\rho\in\TC(\cH)^+$, 
$0 \le \Gamma'\rho \le \Gamma\rho$.
Therefore
$\|\Gamma' \rho \|_{1} \le \|\Gamma \rho \|_{1}$,
by the observation
that \mbox{$0 \le \rho' \le \rho \Rightarrow
  \|\rho'\|_1 \le \|\rho\|_1$} [Because trace and trace norm are equal
for positive operators].
But, for any $\epsilon > 0$, there is $0 < \rho$ such that
\mbox{$\|\Gamma\rho\|\ge (\|\Gamma\|_{1,1}-\epsilon) \|\rho\|_1$}.
\end{proof}
\begin{lem}\label{lem:Lambda rho -> 0}
  $0 \le \rho = \wh{P_s}\rho
  \;\Rightarrow\;
\Lambda_n\rho\xrightarrow{\|\cdot\|_1} 0$.
\end{lem}
\begin{proof}
For fixed $a$, $\wh{P_{a,s}}\Lambda = 0$
for large enough $m$ (see Section \ref{sec:matrix}), hence
also $\wh{P_a}(\Lambda_m\rho) = (\wh{P_{a,s}}\Lambda_m)\rho = 0$.
But, then, eigenvectors of $\Lambda_m\rho$ to nonzero (positive)
eigenvalue are in $P_a^\perp\cK$, and
\begin{equation}\nonumber
0 \le \Lambda_m\rho = \wh{P_a^\perp}(\Lambda_m\rho)
\le \wh{P_a^\perp}(\Lambda \rho).
\end{equation}
But, $\wh{P_a^\perp}(\Lambda \rho)$ tends to zero as $a\to \infty$, 
because its trace is the tail of the convergent series for $\Tr \Lambda\rho$
adapted to the basis $\sB_\cK$.
\end{proof}

\section{Conclusion}

We have given a direct process tomography algorithm for construction
of Kraus decompositions over separable Hilbert spaces.
The proof that it works proves, \textit{a fortiori}, existence of
Kraus decompositions. Being both elementary and constructive,
this is an improvement over the standard proof.
In addition, it is practical for application, and can be useful
for large-but-finite dimension.
With $P_{a,s}$ the polyprojector onto subspaces 
\mbox{$\cH_s \defeq \Span \{\hh_1,\ldots,\hh_s\}$} 
and \mbox{$\cK_a \defeq \Span \{\kk_1,\ldots,\kk_a\}$},
\begin{equation}\label{eq:decomp of projected Lambda}
\wh{P_{a,s}}\Lambda = \sum_{i=1}^{N(a,s)} \btheta(\wh{P_{a,s}}A_i),
\end{equation}
where 
\begin{equation}\nonumber
N(a,s) \defeq \max\{n\,|\, (\gk(n),\gh(n)) \in \cK_a\times\cH_s\}.
\end{equation}
Thus, we may think of the algorithm also as delivering exact
Kraus decompositions of the projections of $\Lambda$ onto bigger
and bigger subspaces, which decompositions are, however, very special
in being maximally compatible among themselves.

\begin{acknowledgments}
This work was supported by NSF MRSEC DMR-2011839.
\end{acknowledgments}
\bibliography{gks}

@article{Chruscinski-22,
Author = {Chruscinski, Dariusz},
Title = {Dynamical maps beyond Markovian regime?},
Journal = {Physics Reports-Review Section of Physics Letters},
Year = {2022},
Volume = {992},
Pages = {1-85},
Month = {DEC 19},
DOI = {10.1016/j.physrep.2022.09.003},
EarlyAccessDate = {OCT 2022},
ISSN = {0370-1573},
EISSN = {1873-6270},
ResearcherID-Numbers = {Chruscinski, Dariusz/H-4509-2011},
Unique-ID = {WOS:000881530600001},
}

@article{GKS,
Author = {Gorini, V and Kossakowski, A and Sudarshan, ECG},
Title = {Completely positive dynamical semigroups of N-level systems},
Journal = {Journal of Mathematical Physics},
Year = {1976},
Volume = {17},
Number = {5},
Pages = {821-825},
DOI = {10.1063/1.522979},
ISSN = {0022-2488},
Unique-ID = {WOS:A1976BP19000033},
}

@article{Grabowski+Kus+Marmo-07,
Author = {Grabowski, Janusz and Kus, Marek and Marmo, Giuseppe},
Title = {On the relation between states and maps in infinite dimensions},
Journal = {Open Systems \& Information Dynamics},
Year = {2007},
Volume = {14},
Number = {4},
Pages = {355-370},
Month = {DEC},
DOI = {10.1007/s11080-007-9061-3},
ISSN = {1230-1612},
EISSN = {1793-7191},
ResearcherID-Numbers = {Kus, Marek/N-3936-2014
   Grabowski, Janusz/K-1248-2013},
ORCID-Numbers = {Kus, Marek/0000-0002-2767-3251
   Grabowski, Janusz/0000-0001-8715-2370},
Unique-ID = {WOS:000252247700002},
}

@book{Breuer+Petruccione,
Author={Breuer, H.-P. and Petruccione, F.},
Title={Theory of Open Quantum Systems},
Publisher={Oxford University Press},
ISBN={978-0-198-52063-4},
Year={2002}
}

@book{Nielsen+Chuang,
Author={Nielsen, M.~A. and Chuang, I.~L.},
Title={Quantum Computation and Quantum Information},
Publisher={Cambridge University Press},
DOI={10.1017/CBO9780511976667},
ISBN={978-0-521-63503-5},
Year={2000}
}

@article{Choi-75,
Author = {Choi, MD},
Title = {Completely Positive Linear Maps on Complex Matrices},
Journal = {Linear Algebra and its Applications},
Year = {1975},
Volume = {10},
Number = {3},
Pages = {285-290},
DOI = {10.1016/0024-3795(75)90075-0},
ISSN = {0024-3795},
Unique-ID = {WOS:A1975AF31800008},
}

@article{Kraus-71,
Author = {Kraus, K},
Title = {General State Changes in Quantum Theory},
Journal = {Annals of Physics},
Year = {1971},
Volume = {64},
Number = {2},
Pages = {311-335},
DOI = {10.1016/0003-4916(71)90108-4},
ISSN = {0003-4916},
Unique-ID = {WOS:A1971J699400001},
}

@book {Kraus-83,
    AUTHOR = {Kraus, Karl},
     TITLE = {States, effects, and operations},
    SERIES = {Lecture Notes in Physics},
    VOLUME = {190},
 PUBLISHER = {Springer-Verlag, Berlin},
      YEAR = {1983},
      ISBN = {3-540-12732-1},
       DOI = {10.1007/3-540-12732-1},
       URL = {https://doi.org/10.1007/3-540-12732-1},
}

@book {Hayashi-17,
    AUTHOR = {Hayashi, Masahito},
     TITLE = {Quantum information theory},
    SERIES = {Graduate Texts in Physics},
   EDITION = {Second},
 PUBLISHER = {Springer-Verlag, Berlin},
      YEAR = {2017},
      ISBN = {978-3-662-49723-4; 978-3-662-49725-8},
       DOI = {10.1007/978-3-662-49725-8},
       URL = {https://doi.org/10.1007/978-3-662-49725-8},
}

@misc{Lidar-20,
      title={Lecture Notes on the Theory of Open Quantum Systems}, 
      author={Daniel A. Lidar},
      year={2020},
      eprint={1902.00967},
      archivePrefix={arXiv},
      primaryClass={quant-ph},
      url={https://arxiv.org/abs/1902.00967},
      DOI={10.48550/arXiv.1902.00967},
}

@book {Vacchini-Foundations,
    AUTHOR = {Vacchini, Bassano},
     TITLE = {Open quantum systems---foundations and theory},
    SERIES = {Graduate Texts in Physics},
 PUBLISHER = {Springer, Cham},
      YEAR = {2024},
      ISBN = {978-3-031-58217-2; 978-3-031-58218-9},
       DOI = {10.1007/978-3-031-58218-9},
       URL = {https://doi.org/10.1007/978-3-031-58218-9},
}

@book {Benatti-09,
    AUTHOR = {Benatti, Fabio},
     TITLE = {Dynamics, information and complexity in quantum systems},
    SERIES = {Theoretical and Mathematical Physics},
 PUBLISHER = {Springer, Berlin},
      YEAR = {2009},
      ISBN = {978-1-4020-9305-0},
MRREVIEWER = {Robert\ Alicki},
       DOI = {10.1007/978-1-4020-9306-7},
       URL = {https://doi.org/10.1007/978-1-4020-9306-7},
}

@misc{Lammert-26,
      title={Notes on completely positive maps and continuous-time Markovian CP evolution. A geometry-flavored perspective}, 
      author={Paul E. Lammert},
      year={2026},
      eprint={2507.11766},
      archivePrefix={arXiv},
      primaryClass={quant-ph},
      url={https://arxiv.org/abs/2507.11766}, 
      DOI={10.48550/arXiv.2507.11766},
}

@article{Chuang+Nielsen-97,
author = {Isaac L. Chuang and M. A. Nielsen},
title = {Prescription for experimental determination of the dynamics of a quantum black box},
journal = {Journal of Modern Optics},
volume = {44},
number = {11-12},
pages = {2455--2467},
year = {1997},
publisher = {Taylor \& Francis},
doi = {10.1080/09500349708231894},
URL = { https://www.tandfonline.com/doi/abs/10.1080/09500349708231894},
eprint = {https://www.tandfonline.com/doi/pdf/10.1080/09500349708231894}
}

@article{Mohseni+Rezakhani+Lidar-08,
	title = {Quantum-process tomography: {Resource} analysis of different strategies},
	volume = {77},
	shorttitle = {Quantum-process tomography},
	url = {https://link.aps.org/doi/10.1103/PhysRevA.77.032322},
	doi = {10.1103/PhysRevA.77.032322},
	number = {3},
	urldate = {2026-07-16},
	journal = {Physical Review A},
	publisher = {American Physical Society},
	author = {Mohseni, M. and Rezakhani, A. T. and Lidar, D. A.},
	month = mar,
	year = {2008},
	pages = {032322},
}

@book {Jaeger-07,
    AUTHOR = {Jaeger, Gregg},
     TITLE = {Quantum information},
      NOTE = {An overview,
              With a foreword by Tommaso Toffoli},
 PUBLISHER = {Springer, New York},
      YEAR = {2007},
     PAGES = {xviii+284},
      ISBN = {978-0-387-35725-6; 0-387-35725-4},
   DOI={10.1007/978-0-387-36944-0},
   MRCLASS = {81P68 (68Q05 81-01)},
  MRNUMBER = {2270789},
MRREVIEWER = {Simon\ J.\ Gay},
}

@article{Jiang+Luo+Fu-13,
  title = {Channel-state duality},
  author = {Jiang, Min and Luo, Shunlong and Fu, Shuangshuang},
  journal = {Phys. Rev. A},
  volume = {87},
  issue = {2},
  pages = {022310},
  numpages = {8},
  year = {2013},
  month = {Feb},
  publisher = {American Physical Society},
  doi = {10.1103/PhysRevA.87.022310},
  url = {https://link.aps.org/doi/10.1103/PhysRevA.87.022310},
}
\bibliographystyle{apsrev4-2}

\end{document}